\documentclass[journal]{IEEEtran}
\usepackage{amsmath, amssymb}
\usepackage{amsthm}
\usepackage{multirow}
\usepackage{url}
\usepackage{graphicx}
\usepackage{subcaption}
\usepackage[ruled,vlined,linesnumbered]{algorithm2e} 
\newtheorem{definition}{Definition}

\newtheorem{theorem}{Theorem}
\newtheorem{lemma}[theorem]{Lemma}
\usepackage{xcolor}

\title{Structurally Valid Log Generation using FSM-GFlowNets}

\author{\IEEEauthorblockN{Riya Samanta}
\IEEEauthorblockA{\textit{Techno India University, Kolkata, India} \\
Email: riya.s@technoindiaeducation.com}}

\begin{document}

\maketitle

\begin{abstract}
% \textcolor{red}{Add a line introducing the problem, stating the problem statement-emphasizing why its a challenging, yet important problem; before coming to a solution approach.}
Generating structurally valid and behaviorally diverse synthetic event logs for interaction-aware models is a challenging yet crucial problem, particularly in settings with limited or privacy-constrained user data. Existing methods—such as heuristic simulations and LLM-based generators—often lack structural coherence or controllability, producing synthetic data that fails to accurately represent real-world system interactions. This paper presents a framework that integrates Finite State Machines (FSMs) with Generative Flow Networks (GFlowNets) to generate structured, semantically valid, and diverse synthetic event logs. Our FSM-constrained GFlowNet ensures syntactic validity and behavioural variation through dynamic action masking and guided sampling. The FSM, derived from expert traces, encodes domain-specific rules, while the GFlowNet is trained using a flow-matching objective with a hybrid reward balancing FSM compliance and statistical fidelity. We instantiate the framework in the context of UI interaction logs using the UIC HCI dataset, but the approach generalizes to any symbolic sequence domain. Experimental results based on distributional metrics  show that our FSM-GFlowNet produces realistic, structurally consistent logs, achieving, for instance, Under the real user logs baseline, a KL divergence of 0.2769 and Chi-squared distance of 0.3522 (significantly outperforming GPT-4o's 2.5294/13.8020 and Gemini's 3.7233/63.0355), alongside a leading bigram overlap of 0.1214 (Vs GPT-4o's 0.0028 and Gemini's 0.0007). A downstream use case—intent classification—demonstrates that classifiers trained solely on our synthetic logs produced from FSM-GFlowNet generalizes well to real-world data, achieving above 77 \% accuracy on real user sessions, validating the framework’s broad applicability.

% \textcolor{red}{highlight the key results in numeric form}

\end{abstract}


\section{Introduction}
\vspace{-0.05in}

Across domains such as Human-Computer Interaction (HCI), cybersecurity, robotic planning, and process mining, system behaviors are naturally expressed as structured sequences of symbolic events \cite{carrera2023exploring,theis2020hci}. These \textit{event logs}—annotated with timestamps, interaction types, and contextual states—form the foundation for tasks like anomaly detection, user modeling, and workflow optimization \cite{pradhan2025getting}. Yet, access to such logs remains severely limited due to privacy regulations, inconsistent instrumentation, and sparse availability. Consequently, synthetic log generation has emerged as a pragmatic alternative to support data-driven interaction modeling in privacy-constrained or low-resource environments \cite{samanta2024ctg}.

Structured user interaction logs, in particular, encode rich behavioral signals such as intent, task strategy, cognitive load, and usability patterns. When semantically coherent and temporally grounded, these logs facilitate adaptive interface design and intent-aware systems \cite{pradhan2025getting}. However, real-world UI logs are often noisy, unlabeled, and fragmented—limiting their utility for model training and evaluation \cite{wang2024open,le2023exploring,samanta2024ctg}. They rarely capture edge cases or optimal trajectories, further underscoring the need for synthetic generation strategies that are both scalable and structure-preserving \cite{balog2025user}.

A \textit{structured UI task} entails a goal-driven sequence of discrete interactions—e.g., navigating folders, editing text, or performing computations—each abstracted as a symbolic (state, action) pair. The resulting \textit{structured symbolic sequence} encodes both procedural logic and interface context \cite{theis2020hci,berkovitch2024identifying}. Finite State Machines (FSMs), Petri Nets, and graph-based models are well-suited for representing these interactions due to their formal expressiveness and traceability \cite{theis2019behavioral,kanade2023fsm}. FSMs, in particular, offer deterministic control over valid transitions and strong guarantees on syntactic correctness—making them ideal for capturing the structural backbone of interaction workflows.

Several approaches have attempted to generate synthetic UI logs, including heuristic rule engines, replay-based scripting \cite{xuagenttrek}, and prompt-based generation using Large Language Models (LLMs) like GPT-4 \cite{confidentai2024synthetic,balog2025user}. While LLMs exhibit impressive sequence modeling capabilities, they often suffer from hallucinations, drift in logical flow, and lack mechanisms for enforcing strict task-level constraints \cite{huang2025survey}. As a result, they struggle to produce logs that are simultaneously coherent, controllable, and semantically aligned with real UI interactions.

FSMs, on the other hand, offer robust structural grounding but are inherently non-generative—they can validate sequences but cannot explore behavioral variation or synthesize new trajectories \cite{kanade2023fsm}. This presents a critical gap: \textit{how can we generate diverse synthetic UI logs that are both structurally valid and semantically plausible?}

We address this challenge by introducing a hybrid framework that integrates FSMs with Generative Flow Networks (GFlowNets)—a class of probabilistic models designed for sampling structured objects with probability proportional to a learned reward \cite{krichel2024generalization,deleu2022bayesian}. In our \textit{FSM-Constrained GFlowNet}, the FSM defines permissible state transitions, which are used to dynamically mask the GFlowNet action space at each timestep. This ensures that every sampled trajectory strictly adheres to valid task logic while still supporting diverse generative exploration through flow-matching. The FSM is reverse-engineered from expert-level trajectories synthesized using GPT-4o, and training is guided by a hybrid reward function balancing FSM compliance with statistical fidelity to real-world logs from the UIC HCI dataset \cite{theis2020hci,theis2019behavioral}. The resulting synthetic logs exhibit strong structural and behavioral alignment with authentic interactions, significantly outperforming unconstrained LLM baselines across multiple distributional metrics. Moreover, we demonstrate their downstream utility in training intent classifiers that generalize to real user data.

% \vspace{-0.1in}
% \subsection*{Contributions}
% \vspace{-0.01in}
Our primary contributions are summarized as follows:
\begin{itemize}
   \item We present a framework that integrates Finite State Machines with Generative Flow Networks to enable the controlled synthesis of structurally valid and semantically rich symbolic event sequences, and instantiate it in the domain of user interface interaction modeling using FSMs derived from expert GPT-4o traces and real user data from the UIC HCI dataset.
   \item We evaluate our method against LLM-based log generation using quantitative distributional metrics, including KL divergence, Chi-squared distance, entropy, and bigram overlap, under both real and ground truth baselines.
   \item We validate the practical utility of our generated logs via a downstream use case that is intent classification, where models trained solely on FSM-GFlowNet logs achieve competitive performance on real user logs.
\end{itemize}

The remainder of the paper is organized as follows. Section~\ref{sec:related} reviews related work. Section~\ref{data} describes the dataset. Section~\ref{sec:meth} presents the proposed FSM-GFlowNet framework followed by theoretical validations in Section~\ref{sec:th} Section~\ref{sec:eval} reports experimental results. Section~\ref{sec:use} discusses a downstream use case. Section~\ref{sec:con} concludes the paper.

\section{Related Work}
\label{sec:related}
Synthetic log generation for human-machine interaction remains an emerging area, with methods ranging from rule-based simulations to advanced machine learning and generative models. Below, we categorize the key approaches and identify critical gaps that motivate our proposed solution.

\subsection*{Traditional Log Generation}

Early methods for synthetic UI log generation employed rule-based simulations or replay engines that mimic recorded behavior. Martinez-Rojas et al.~\cite{martinez2022tool} introduced a tool-supported method to extend logs using templates and event variability functions, but these logs are limited in scope and generalizability. Replay-based scripting methods are often used in Robotic Process Automation (RPA) pipelines to mimic user tasks, but they remain restricted to pre-recorded behaviors and do not capture unseen task paths~\cite{reijers2023rpa}. These methods fall short in producing dynamic, long-horizon logs with realistic human-like variability.

\subsection*{Learning-Based Synthetic Log Generation}
 
In industrial contexts, Generative Adversarial Networks (GANs) have been used to simulate log files for anomaly detection in cyber-physical systems. Partovian et al.~\cite{partovian2024leveraging} use CTGAN and SeqGAN to generate logs for smart-troubleshooting, but these methods primarily target tabular anomaly enrichment and are not optimized for structured behavioral sequence modeling. Separately, the CTG-KrEW framework~\cite{samanta2024ctg} improves semantic correlation in skill-oriented tabular data, yet it does not model sequential dependencies or application state transitions critical to UI tasks.

Recurrent Neural Networks (RNNs) and Long Short-Term Memory (LSTM) models have shown potential in generating synthetic sequential logs in fields such as petroleum engineering. Zhang et al.~\cite{zhang2018synthetic} proposed using LSTMs to reconstruct missing well logs, demonstrating their capacity to preserve temporal dependencies. However, these approaches assume homogenous time-series structures and lack symbolic constraints, which are essential for modeling interface state machines in HCI systems.

\subsection*{LLM-Based Log Synthesis and Logging}

Large Language Models (LLMs) have recently been adopted for structured task generation. UniLog~\cite{xu2024unilog} demonstrates that GPT-based models can generate code logging statements via in-context learning, showing impressive performance in automated logging. Similarly, Li et al.~\cite{li2024exploring} evaluated LLMs for software log statement generation using LogBench and showed their advantage in semantic comprehension. However, these models often suffer from hallucinations, lack precise control over transitions, and cannot enforce application-specific rules~\cite{huang2025survey}. Our approach addresses this limitation by using LLMs only to derive optimal FSM paths, not for full sequence generation.

\subsection*{Novelty of Our Approach}
Finite State Machines (FSMs) have traditionally been employed in HCI to model deterministic workflows and enforce structural validity~\cite{kanade2023fsm,king2024thoughtful}. On the other hand, Generative Flow Networks (GFlowNets)~\cite{deleu2022bayesian,krichel2024generalization} represent a recent probabilistic modeling paradigm aimed at sampling structured objects under reward-driven constraints. While GFlowNets have demonstrated promising results in molecular design and combinatorial generation, their application to human-machine interaction or FSM-constrained sequence generation remains unexplored.

To the best of our knowledge, no existing work systematically integrates FSM-driven symbolic constraints within a GFlowNet framework for FSM-constrained symbolic sequence generation, with UI log generation as a core case study. Our approach fills this critical gap by dynamically constraining the GFlowNet's action space based on FSM transitions, enabling the synthesis of structurally valid, stochastically diverse, and semantically meaningful synthetic logs. This hybrid framework supports scalable generation suitable for downstream analysis tasks, such as user intent discovery and behavioral modeling in human-machine systems.

\begin{table*}[!t]
\centering
\caption{Structure of Log Files in the \texttt{simple} Folder}
\label{log}
\begin{tabular}{|c|l|}
\hline
\textbf{Field} &  \textbf{Description} \\
\hline
Timestamp  & Time of the event in milliseconds \\
\hline
Action Code & Encoded event types (e.g., mouse move, app switch, close window) \\
\hline
Metadata & Event-specific context: mouse coordinates or window details\\
\hline
\end{tabular}
\vspace{-0.1in}
\end{table*}

\section{Dataset}
\label{data}
We utilize the \textit{Human-Computer Interaction Logs} dataset released by the University of Illinois Chicago (UIC HCI)\cite{theis2020hci}, which records detailed interaction traces of ten real participants performing predefined data processing tasks in a controlled Windows environment using only three standard applications: Notepad, Calculator, and File Explorer. The dataset is organized into two task categories—\textit{Simple} and \textit{Complex}—each containing five anonymized log files corresponding to five different users.

In this study, we focus exclusively on the \textit{Simple} category, which is characterized by clearly defined, structured sequences of interactions. In this task, each participant was provided with 30 raw text files located in the `Documents/CompanyData' folder. Each file contained revenue and expense information for a single product over a certain time period. Participants were instructed to generate 15 summaries by pairing up two files at a time, computing the combined revenue, expenses, and profit for each pair using Calculator, and recording the results into a `summary.txt' file using Notepad. The summary file was saved inside the `Documents/CompanyData/Summaries' directory. All navigation between folders and file access operations were to be done strictly via File Explorer, adhering to the application constraints.

The dataset files are formatted as timestamped logs, where each row represents a user-generated event such as a keyboard stroke, mouse movement, or application context switch, along with additional metadata including window identifiers and screen coordinates. Throughout this manuscript, we refer to this complete sequence of user interaction—from file browsing to summary writing as the \textbf{Task}. The structural format of the log files used from the \textit{Simple} category is detailed in Table \ref{log}.

\section{Methodology}
\label{sec:meth}
Our proposed framework, FSM-GFlowNet, is a general-purpose approach for generating structured, semantically valid, and behaviorally diverse symbolic event sequences. It integrates Finite State Machines (FSMs), which encode domain-specific structural rules, with Generative Flow Networks (GFlowNets), which enable stochastic, reward-driven trajectory sampling. The objective is to synthesize long and coherent interaction sequences that obey logical constraints while supporting diversity and generalization—applicable across domains such as UI interaction modeling, robotic planning, and event-driven system simulation.

In this work, we instantiate the framework for the generation of synthetic user interface (UI) interaction logs to demonstrate its effectiveness in a realistic and behaviorally grounded setting. The methodology comprises three core components that together ensure structural correctness, semantic richness, and controlled generation: (1) LLM-Based Optimal Path Construction, (2) Finite State Machine Modeling, and (3) FSM-Constrained GFlowNet-Based Log Generation.

To assess the quality of the generated sequences, we conduct a multi-faceted evaluation. First, we compare FSM-GFlowNet logs with both real user logs and unconstrained LLM-generated logs using distributional metrics such as KL divergence, Chi-squared distance, entropy, and bigram overlap. Second, we validate the practical value of the generated data through a downstream task—intent classification—testing whether models trained solely on synthetic logs can generalize to real user logs.

\subsection{LLM-Based Optimal Path Construction}
\label{subsec:gt-construction}

Given the limited scale, lack of labelled data , and abstract interaction semantics inherent in the UIC HCI dataset~\cite{theis2019behavioral}, applying supervised learning approaches to model user behavior is not practical. To overcome this, we adopt a prompt-based engineering strategy~\cite{berkovitch2024identifying} using state-of-the-art Large Language Models (LLMs), particularly GPT-4o~\cite{berti2023abstractions,petruzzellis2024benchmarking,gpt4o_wikipedia}, to synthesize an expert-level, semantically coherent UI interaction trajectory.

The objective of this step is to construct a canonical, efficient, and error-free sequence of UI actions that a proficient user would perform to complete the structured data processing task using the core applications mentioned in the \textbf{Task}. This high-quality LLM-generated trajectory is referred to as the \textit{Ground Truth (GT) path} and fulfills a dual role within our framework.

First, it provides a reference for \textbf{reverse-engineering a Finite State Machine (FSM)} that encodes valid UI state transitions and interaction constraints. This FSM is later used to constrain the trajectory sampling in our proposed FSM-guided GFlowNet. The FSM design ensures that all sampled trajectories are structurally valid and aligned with task-specific logic. Second, and equally importantly, the GT path acts as a \textbf{evaluation benchmark during evaluation}. As detailed in Section~\ref{sec:eval}, the GT log is included alongside real and LLM-generated baselines for assessing the statistical and structural alignment of synthetic UI logs. This allows us to quantify how well each generation method approximates the idealized task execution from a planning and transition-efficiency perspective.

To ensure reproducibility and minimize hallucinations often observed in open-ended generative outputs~\cite{huang2025survey}, we formulate a task-specific prompt encoding the task description, application constraints, expected file manipulations, and UI schema. The prompt is designed to elicit deterministic, context-aware sequences that adhere to practical software behaviors. The complete prompt used for optimal path synthesis (or GT) is shown in Figure~\ref{fig:prompt}.

\begin{figure*}[!t]
    \centering
    \includegraphics[width=0.9\linewidth]{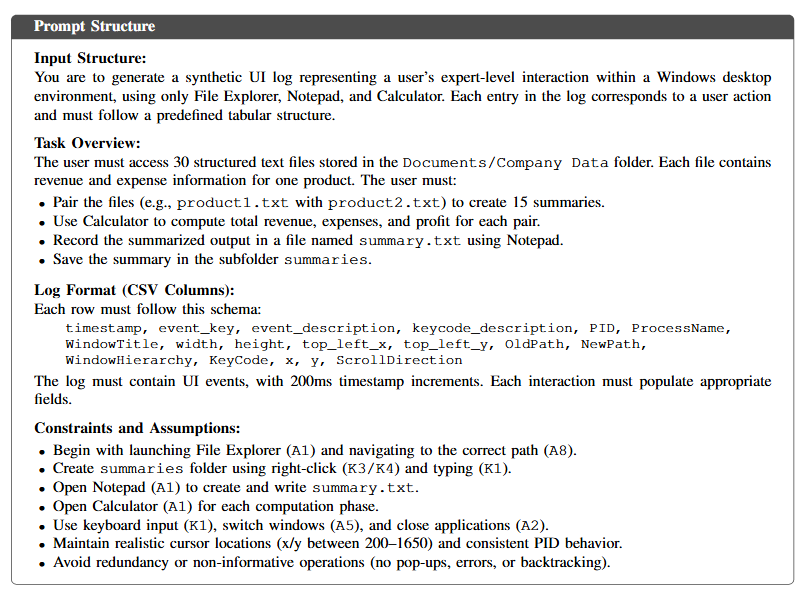}  
  \caption{Prompt design for generating optimal user interaction logs to establish the benchmark ground truth for the \textbf{Task}.}
    \label{fig:prompt}
    \vspace{-0.1in}
\end{figure*}

\subsection{Reverse-engineered Finite State Machine (FSM) Modeling}
To formalize valid UI interaction flows and enforce structural correctness during synthetic log generation, we construct a Finite State Machine (FSM) based on the optimal interaction trajectory synthesized through LLM prompting. Rather than assuming a predefined model, we reverse-engineer the FSM by analyzing the LLM-generated sequence and segmenting it into semantically meaningful application contexts and event transitions.

The FSM serves as a behavioural scaffold that captures permissible transitions among discrete UI contexts—such as navigating folders, opening applications, typing in Notepad, and performing calculations—based on both application semantics and realistic usage patterns observed in the benchmark path.

We define the FSM as a 5-tuple:

\[
\mathcal{M} = (Q, \Sigma, \delta, q_0, F)
\]
\begin{IEEEeqnarray*}{rCl}
Q      &=& \{ S_1, S_2, S_3, S_4 \} \quad \text{(application contexts)} \\
\Sigma &=& \{ A1, A2, A5, A7,\\*
&&\qquad A8, K1, K3, K4, M \} \quad \text{(interaction events)} \\
\delta &:& Q \times \Sigma \rightarrow Q \quad \text{(transition function)} \\
q_0    &=& S_1 \quad \text{(initial state)} \\
F      &=& \{ \bullet \} \quad \text{(terminal state)}
\end{IEEEeqnarray*}

\noindent The transition function \( \delta \) is defined as:
\begin{IEEEeqnarraybox}[\IEEEeqnarraystrutmode\IEEEeqnarraystrutsizeadd{2pt}{2pt}]{l}
\delta(S_1, A8) = S_2, \quad
\delta(S_1, A1) \in \{S3,S4\}, \quad
\delta(S_1, M) = S_1, \\
\delta(S_1, A2) \in F \\
\delta(S_2, A1) \in \{S_3, S_4\}, \quad
\delta(S_2, A8) = S_1, \\
\delta(S_2, e) = S_2, \ \forall e \in \{K3, K4, M\}, \\
\delta(S_3, A1) = S_4, \quad
\delta(S_3, A2) = S_1, \\ 
\delta(S_3, e) = S_3, \ \forall e \in \{K1, M\}, \\
\delta(S_4, A2) = S_1, \quad
\delta(S_4, e) = S_4, \ \forall e \in \{K1, M\}
\end{IEEEeqnarraybox}

% \vspace{-0.05in}
To ensure semantic accuracy and domain relevance, the FSM was validated and refined with input from two HCI experts familiar with Windows-based UI conventions. Their feedback helped clarify transition ambiguities, consolidate equivalent states, and verify behavioral plausibility. The final FSM is illustrated in Figure~\ref{fig:fsm-diagram}, where nodes represent abstracted application contexts and edges denote allowed transitions driven by specific interaction events.

\begin{figure}[!t]
    \centering
    \includegraphics[width=0.40\textwidth]{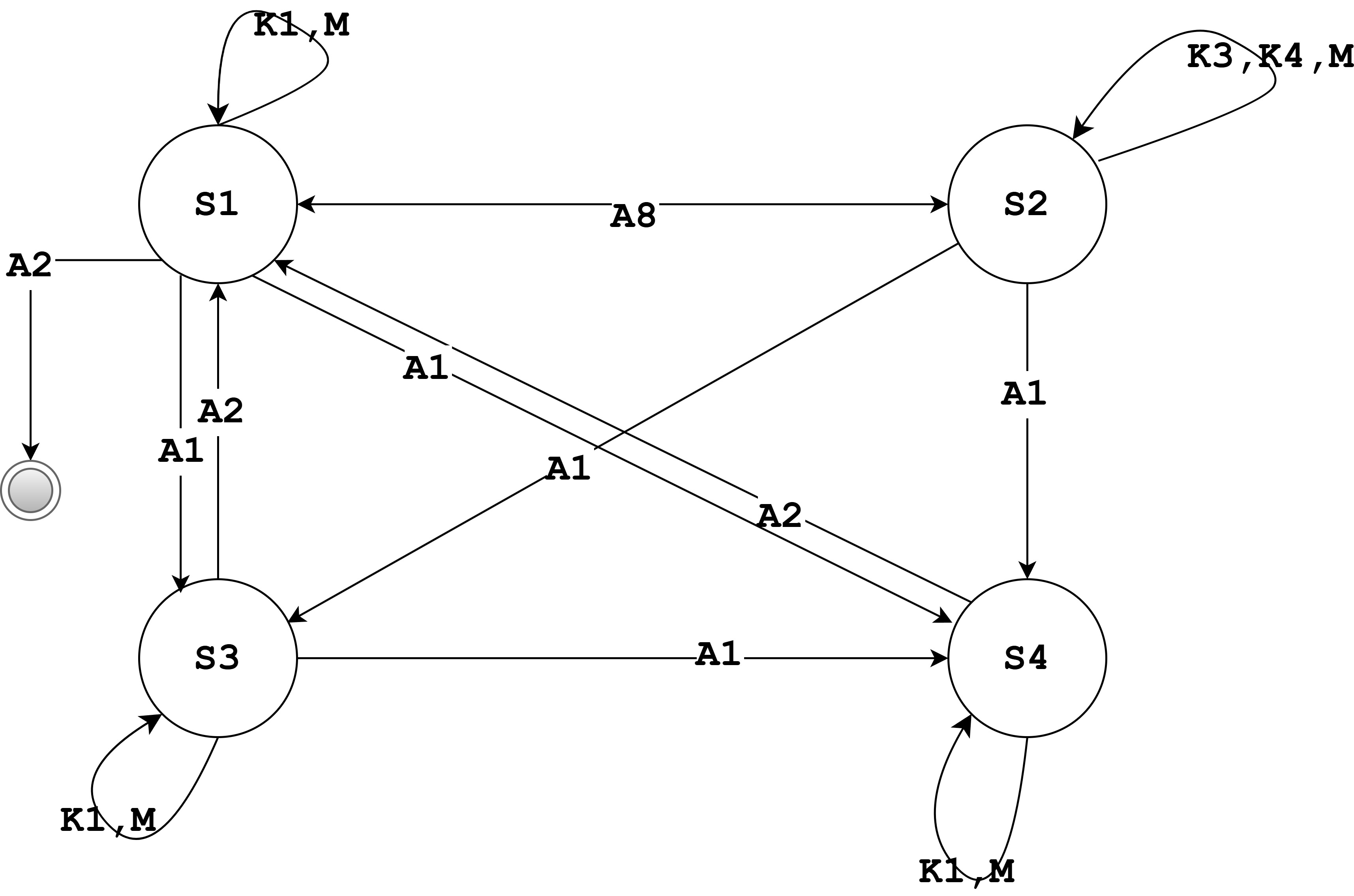}
    \caption{Formalized FSM derived from the benchmark ground truth log. States correspond to contextual application views; transitions are governed by semantically valid interaction events.}
    \label{fig:fsm-diagram}
    \vspace{-0.1in}
\end{figure}

\begin{table*}[!t]
\centering
\caption{State Descriptions, Actions, and Transitions in the FSM}
\label{state}
\begin{tabular}{|c|l|l|l|}
\hline
\textbf{State} & \textbf{Description} & \textbf{Actions} & \textbf{Transitions To} \\
\hline
\texttt{S1} & File Explorer Open & 
\texttt{K1}, \texttt{M} (Remain), \texttt{A8} (to S2), \texttt{A1} (to S3/S4), \texttt{A2} (to Terminal) & 
\texttt{S1}, \texttt{S2}, \texttt{S3}, \texttt{S4}, \texttt{Terminal} \\
\hline
\texttt{S2} & File Explorer Navigating & 
\texttt{K3}, \texttt{K4}, \texttt{M} (Remain), \texttt{A1} (to S3/S4), \texttt{A8} (to S1) & 
\texttt{S2}, \texttt{S1}, \texttt{S3}, \texttt{S4} \\
\hline
\texttt{S3} & Notepad Open & 
\texttt{K1}, \texttt{M} (Remain), \texttt{A1} (to S4), \texttt{A2} (to S1) & 
\texttt{S3}, \texttt{S1}, \texttt{S4} \\
\hline
\texttt{S4} & Calculator Open & 
\texttt{K1}, \texttt{M} (Remain), \texttt{A2} (to S1), \texttt{A1} (to S3) & 
\texttt{S4}, \texttt{S1}, \texttt{S3} \\
\hline
\texttt{Terminal} & End State & -- & -- \\
\hline
\end{tabular}
\end{table*}

\begin{table*}[!t]
\centering
\caption{Descriptions of Actions Used in FSM}
\label{actions}
\begin{tabular}{|c|l|}
\hline
\textbf{Action} & \textbf{Description} \\
\hline
\texttt{A1} & Switch to another application (e.g., from File Explorer to Notepad or Calculator) \\
\hline
\texttt{A2} & Exit current application or transition to Terminal (end state) \\
\hline
\texttt{A8} & Change view or navigate within File Explorer (e.g., from S1 to S2 or vice versa) \\
\hline
\texttt{K1} & Keyboard input within Notepad, Calculator, or File Explorer main screen \\
\hline
\texttt{K3/K4} & Keyboard or mouse-based file/folder navigation within File Explorer navigation view \\
\hline
\texttt{M} & Mouse hover action, typically does not cause state change \\
\hline
\end{tabular}
\vspace{-0.1in}
\end{table*}

\subsection{FSM-Constrained GFlowNet-Based Log Generation}

We formulate the problem of UI log synthesis as a trajectory sampling task over a finite set of semantically valid sequences, governed by an FSM. Let the FSM be defined as $\mathcal{M} = (Q, \Sigma, \delta, q_0, F)$, where $Q$ is the set of abstract application contexts (states), $\Sigma$ is the set of allowed actions (UI events), $\delta: Q \times \Sigma \rightarrow Q$ is the deterministic transition function, $q_0$ is the initial state, and $F$ denotes the terminal state(s).

We aim to learn a generative model that samples trajectories $\tau = \{(s_0, a_0), (s_1, a_1), \ldots, (s_T, a_T)\}$ such that:
\[
P_\theta(\tau) \propto R(\tau)
\]
where $P_\theta$ is the trajectory distribution induced by the GFlowNet policy and $R(\tau)$ is a trajectory-level reward function.

Unlike Markov Decision Processes (MDPs) or standard reinforcement learning (RL), where long-horizon credit assignment is difficult, GFlowNets offer a principled way to sample structured objects (trajectories) with probability proportional to a reward. However, applying GFlowNets directly to UI logs risks invalid transitions. Hence, we embed FSM logic as a constraint on the GFlowNet action space at every timestep. This ensures syntactic and semantic correctness throughout generation \cite{pan2023better}.

\subsubsection*{State and Action Representation}

At each time step $t$, the current state $s_t \in Q$ is encoded as:
\[
\phi(s_t, t) = \left[ \mathbf{1}_{s_t}; \frac{t}{T_{\text{max}}} \right] \in \mathbb{R}^{|Q|+1}
\]
where $\mathbf{1}_{s_t}$ is the one-hot encoding of the current FSM state and $t/T_{\text{max}}$ is a normalized time-depth scalar to provide position-awareness.

The policy network $\pi_\theta$ is a function approximator (a two-layer feedforward neural network) that outputs logits over the action space:
\[
\pi_\theta(a_t \mid s_t) = \text{softmax}(f_\theta(\phi(s_t, t)) + \log \mathbf{m}_{s_t})
\]
Here, $\mathbf{m}_{s_t} \in \{0, 1\}^{|\Sigma|}$ is a binary mask applied element-wise to enforce FSM-valid transitions:
\[
\mathbf{m}_{s_t}[i] = 
\begin{cases}
1, & \text{if } a_i \in \Sigma \text{ and } \delta(s_t, a_i) \text{ is defined} \\
0, & \text{otherwise}
\end{cases}
\]

\subsubsection*{Flow Matching Objective}
In standard policy gradient RL, the agent is optimized to maximize expected cumulative reward \cite{wang2022policy}. However, in structured domains like UI logs, reward signals are sparse and delayed (only at termination). GFlowNet overcomes this by enforcing flow conservation.

Following the GFlowNet formulation~\cite{krichel2024generalization,deleu2022bayesian,dutta2023toward}, we define a forward flow $F_\theta(s \rightarrow s')$ and a backward flow $B_\theta(s' \rightarrow s)$ for each transition. The training objective is to match the total incoming and outgoing flow at every non-terminal state $s$:
\[
\sum_{s': (s, a) \rightarrow s'} F_\theta(s \rightarrow s') = \sum_{s'': s'' \rightarrow s} B_\theta(s'' \rightarrow s)
\]

To optimize this, we use a trajectory-based loss:
\[
\mathcal{L}(\theta) = \mathbb{E}_{\tau \sim \pi_\theta} \left[-R(\tau) \cdot \sum_{t=0}^{T} \log \pi_\theta(a_t \mid s_t) \right]
\]

\subsubsection*{Reward Function and Termination Constraints}

The trajectory-level reward is defined to promote longer, FSM-compliant sequences that terminate correctly:
\[
R(\tau) = 
\begin{cases}
\log(|\tau| + 1), & \text{if } s_T \in F \text{ and } \forall t: \delta(s_t, a_t) \text{ is defined} \\
0, & \text{otherwise}
\end{cases}
\]

We further extend the generation process by injecting stochastic \texttt{M} (mouse-hover) actions before or after state-changing events with a fixed probability (e.g., $p=0.4$). This models fine-grained human interaction behavior and increases trajectory diversity without impacting FSM compliance.

\subsubsection*{Sampling and Exploration Strategy}

During generation, actions are sampled via a categorical distribution parameterized by the masked logits. To avoid convergence to trivial or repetitive trajectories, we incorporate an $\varepsilon$-greedy exploration policy:
\[
a_t \sim 
\begin{cases}
\text{Uniform}(\{a \in \Sigma \mid \delta(s_t, a) \text{ is valid}\}), & \text{with probability } \varepsilon \\
\pi_\theta(\cdot \mid s_t), & \text{otherwise}
\end{cases}
\]

This FSM-constrained GFlowNet architecture provides the best of both worlds: structural validity imposed by symbolic FSMs and generative diversity enabled by stochastic policy learning. The generated logs not only align with realistic application logic but also offer sufficient variability for downstream behavior modeling tasks. The procedural steps of the proposed FSM-Constrained GFlowNet framework, encompassing both model training and synthetic log generation, are systematically described in Algorithm~\ref{alg:fsm-gflownet}.

While this implementation focuses on UI interactions, the framework is modular: any domain where FSMs define valid transitions over symbolic events can instantiate the same training and sampling pipeline.

\begin{algorithm}[!t]
\caption{FSM-Constrained GFlowNet}
\label{alg:fsm-gflownet}
\SetAlgoLined
\KwIn{FSM $\mathcal{M} = (Q, \Sigma, \delta, q_0, F)$, episodes $E$, maximum steps $T_{\text{max}}$, model $\pi_\theta$, reward function $R(\tau)$}
\KwOut{Trained GFlowNet model, synthetic UI logs}

\textbf{Training Phase:}\;
Initialize model $\pi_\theta$ with random weights\;
\For{each episode $e \in \{1, \ldots, E\}$}{
    Initialize trajectory $\tau \leftarrow \{\}$\;
    Set initial state $s_0 \leftarrow q_0$, $t \leftarrow 0$\;
    
    \While{$s_t \notin F$ and $t < T_{\text{max}}$}{
        Encode state $\phi(s_t, t)$ (one-hot state + normalized step scalar)\;
        Obtain valid action mask $\mathbf{m}_{s_t}$ from FSM\;
        Compute masked logits: $l \leftarrow \pi_\theta(\phi(s_t, t)) + \log(\mathbf{m}_{s_t} + \epsilon)$\;
        
        Sample action $a_t$ using $\varepsilon$-greedy policy\;
        Append $(s_t, a_t)$ to trajectory $\tau$\;
        
        Update state $s_{t+1} \leftarrow \delta(s_t, a_t)$\;
        Increment step counter $t \leftarrow t+1$\;
        
        Optionally insert mouse hover (\texttt{M}) actions stochastically\;
    }
    Compute reward $R(\tau)$\;
    Update $\pi_\theta$ by minimizing loss $\mathcal{L} = -R(\tau) \sum_{t=0}^{|\tau|} \log \pi_\theta(a_t \mid s_t)$\;
}

\textbf{Log Generation Phase:}\;
\For{each required synthetic log}{
    Initialize $s_0 \leftarrow q_0$, $t \leftarrow 0$, empty log $\mathcal{L}$\;
    \While{length of $\mathcal{L}$ $<$ desired events}{
        Optionally insert \texttt{M} event with probability $p_{\text{hover}}$\;
        Encode state $\phi(s_t, t)$ and obtain valid action mask\;
        Sample next action $a_t$ from $\pi_\theta$ with FSM masking\;
        Append $(s_t, a_t)$ to $\mathcal{L}$\;
        Update $s_{t+1} \leftarrow \delta(s_t, a_t)$ (or reset to $q_0$ if terminal)\;
        Increment $t$\;
    }
    Save log $\mathcal{L}$ to disk in CSV format\;
}

\end{algorithm}

\section{Theoretical Analysis}
\label{sec:th}
\vspace{-0.01in}

We formally analyze the computational complexity and structural guarantees of the proposed FSM-constrained GFlowNet framework for synthetic UI log generation.

\begin{definition}
\textbf{Training Complexity:} Let $E$ denote the number of training episodes, $T_{\text{max}}$ the maximum number of steps per episode, $|Q|$ the number of FSM states, $|\Sigma|$ the number of possible actions, and $H$ the hidden dimension size of the policy network. Then, the computational complexity of training the FSM-constrained GFlowNet is given by:
\begin{equation}
\mathcal{O}\left(E \times T_{\text{max}} \times (|Q| + H + |\Sigma|)\right)
\end{equation}
\end{definition}

\begin{definition}
\textbf{Log Generation Complexity:} Let $N$ denote the number of synthetic logs generated, and $T_{\text{log}}$ the approximate number of events per log. The complexity of generating synthetic logs using the trained policy is given by:
\begin{equation}
\mathcal{O}\left(N \times T_{\text{log}} \times (|Q| + H + |\Sigma|)\right)
\end{equation}
\end{definition}

Given that $|Q|$ and $|\Sigma|$ are typically small and $H$ is moderate for task-specific FSMs, the overall framework remains computationally efficient and scalable for large-scale log generation.

\begin{lemma}
Every trajectory $\tau$ generated by the FSM-constrained GFlowNet strictly adheres to the valid transitions defined by the FSM $\mathcal{M}$.
\end{lemma}

\begin{proof}
At each generation step $t$, the action space is dynamically masked based on the FSM transition function $\delta(s_t, a_t)$. Only actions that satisfy $\delta(s_t, a_t)$ being defined are permitted. Consequently, each transition $(s_t, a_t) \to s_{t+1}$ within a trajectory $\tau$ is valid according to the FSM $\mathcal{M}$. Thus, the generated trajectory remains structurally correct throughout the sampling process.
\end{proof}

\vspace{-0.3in}
\section{Comparative Evaluation}
\label{sec:eval}
\vspace{-0.1in}

To rigorously assess the effectiveness of our proposed FSM-guided GFlowNet framework for UI log synthesis, we conduct a comparative evaluation against three baselines: (i) real-world user logs from the UIC HCI dataset~\cite{theis2020hci}, (ii) unconstrained synthetic logs generated by large language models (LLMs) using prompt in Figure~\ref{fig:prompt-b}, and (iii) a Ground Truth (GT) log derived from GPT-4o using prompt in Figure~\ref{fig:prompt}.

\subsection{Baselines}
\subsubsection*{LLM-Based Synthetic Logs}
We construct our first baseline using two pretrained LLMs: \textbf{ChatGPT-4o (OpenAI)}~\cite{gpt4o_wikipedia} and \textbf{Gemini 1.5 7B (Google)}~\cite{gemma7b_huggingface}. Both models are prompted with an identical natural language description of the UI task (see Figure~\ref{fig:prompt-b}), which outlines the task goal, constraints on software usage (File Explorer, Notepad, Calculator), and output format. These prompts are intentionally unconstrained, allowing the models to simulate user behaviors without enforcing procedural correctness.

We generate 100 logs from each method, with each log containing approximately 1000–1500 UI events. To ensure consistency and fairness—especially considering the UIC HCI dataset comprises only five real user logs—we randomly sample five logs per method per evaluation run. This process is repeated across 100 independent iterations, and all reported results are averaged to ensure statistical robustness. All logs are parsed into structured CSVs and processed using a unified evaluation pipeline applied consistently across FSM-GFlowNet and LLM-based outputs.

\begin{figure*}[!t]
    \centering
    \includegraphics[width=0.9\linewidth]{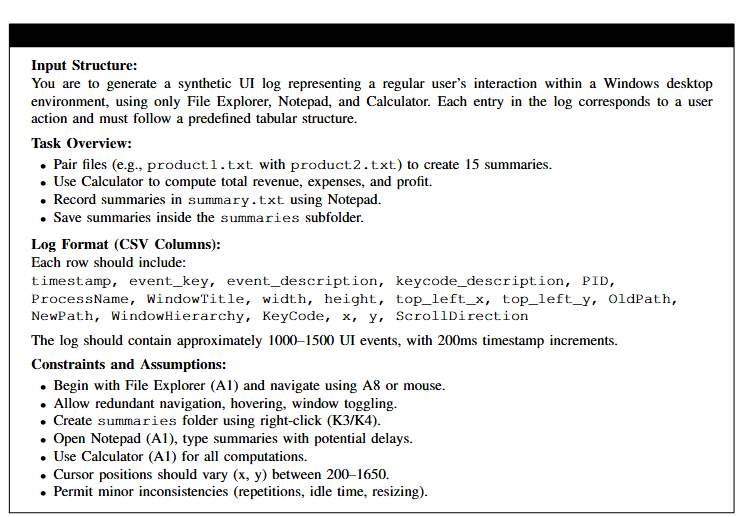}  
  \caption{{Prompt used for LLM-based generation of baseline synthetic UI logs. Designed to simulate a non-expert user with realistic variability.}}
    \label{fig:prompt-b}
    \vspace{-0.1in}
\end{figure*}

\subsubsection*{Real User Logs}
The second baseline consists of actual user interaction traces from the UIC HCI logs dataset~\cite{theis2020hci} as discussed in the Section~\ref{data}.

\subsubsection*{Ground Truth (GT) Log}
The third baseline is a single high-quality trajectory synthesized via prompt engineering with GPT-4o, as detailed in Section~\ref{subsec:gt-construction}. This \textit{Ground Truth (GT)} log represents an \textbf{optimal task} execution: it is deterministic, semantically valid, and conforms strictly to the FSM-derived task logic. While lacking in user-level diversity, the GT serves as a semantic benchmark for evaluating how closely generation methods approximate expert-level action planning.

\subsection{Data Preprocessing}
Prior to quantitative evaluation, all real and synthetic UI logs underwent a standardized cleaning procedure to ensure consistency and comparability across sources. For each log file, we retained only the fields critical to structural evaluation: the user interface \texttt{state} and the associated \texttt{event} type.

Specifically, the \texttt{event} field was processed by extracting only the high-level action categories, such as \texttt{A1} and \texttt{K3}, while discarding finer-grained descriptive information. This extraction isolates the abstracted user interactions necessary for modeling and evaluation. Furthermore, all auxiliary metadata fields, including timestamps, window dimensions, process identifiers, and cursor positions, were removed from the logs. Only the \texttt{state} and \texttt{event} columns were preserved for subsequent analysis.

This cleaning process ensures that all evaluation metrics focus solely on the dynamics of action-state sequences, eliminating any confounding factors arising from low-level UI fluctuations. By standardizing the representation across real and synthetic logs, we enable fair and focused assessment of structural fidelity, behavioral diversity, and generative realism.

\begin{table*}[!t]
\centering
\caption{Quantitative comparison (aggregated distribution) of FSM-GFlowNet and baselines evaluated under two references: (a) real user logs, and (b) the Ground Truth (GT) trajectory.}

\label{tab:quantitative-evaluation}
\renewcommand{\arraystretch}{1.2}
\setlength{\tabcolsep}{15pt}

\textbf{(a) Real User Logs (UIC HCI)}
\vspace{2mm}

\begin{tabular}{|l|c|c|c|c|}
\hline
\textbf{Dataset} & \textbf{KL Divergence} & \textbf{Chi-squared} & \textbf{Entropy} & \textbf{Bigram Overlap} \\
\hline
Real                 & 0.0088  & 0.0175  & 1.3198 & 0.9679 \\
GPT-4o               & 2.5294  & 13.8020 & 0.4445 & 0.0028 \\
Gemini               & 3.7233  & 63.0355 & 1.0417 & 0.0007 \\
Ground Truth (GT)    & 1.5648  & 4.5594  & 1.5892 & 0.0061 \\
\textbf{FSM-GFlowNet} & \textbf{0.2769} & \textbf{0.3522} & \textbf{0.5979} & \textbf{0.1214} \\
\hline
\end{tabular}

\vspace{6mm}
\textbf{(b) Ground Truth (GT)}
\vspace{2mm}

\begin{tabular}{|l|c|c|c|c|}
\hline
\textbf{Dataset} & \textbf{KL Divergence} & \textbf{Chi-squared} & \textbf{Entropy} & \textbf{Bigram Overlap} \\
\hline
Real                 & 14.8965 & $4.59 \times 10^8$ & 1.3198 & 0.9679 \\
GPT-4o               & 0.6487  & 0.9473             & 0.4445 & 0.0028 \\
Gemini               & 2.1526  & $3.00 \times 10^5$ & 1.0417 & 0.0007 \\
Ground Truth (GT)     & 0.0000  & 0.0000             & 1.5892 & 0.0061 \\
\textbf{FSM-GFlowNet} & \textbf{16.9481} & \textbf{$6.86 \times 10^8$} & \textbf{0.5979} & \textbf{0.1214} \\
\hline
\end{tabular}
\end{table*}

\begin{figure*}[!t]
\centering
\label{figs-real}
\captionsetup[subfigure]{labelformat=empty}

\begin{subfigure}[t]{0.45\textwidth}
    \centering
    \includegraphics[width=\textwidth]{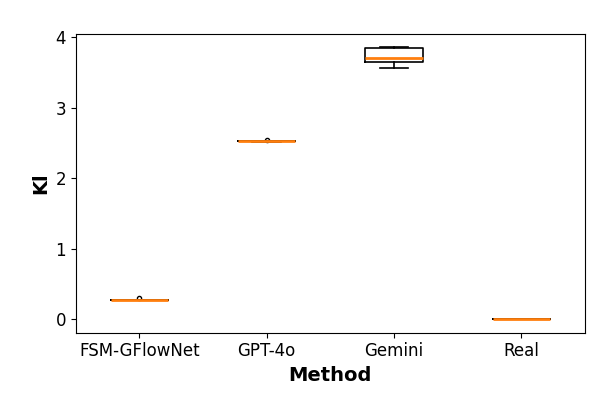}
    \caption{\textbf{(a)} KL Divergence}
\end{subfigure}
\hfill
\begin{subfigure}[t]{0.45\textwidth}
    \centering
    \includegraphics[width=\textwidth]{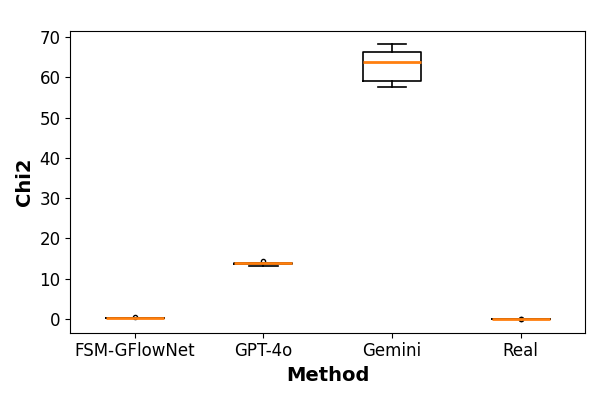}
    \caption{\textbf{(b)} Chi-Squared Distance}
\end{subfigure}

% \vspace{1mm}

\begin{subfigure}[t]{0.45\textwidth}
    \centering
    \includegraphics[width=\textwidth]{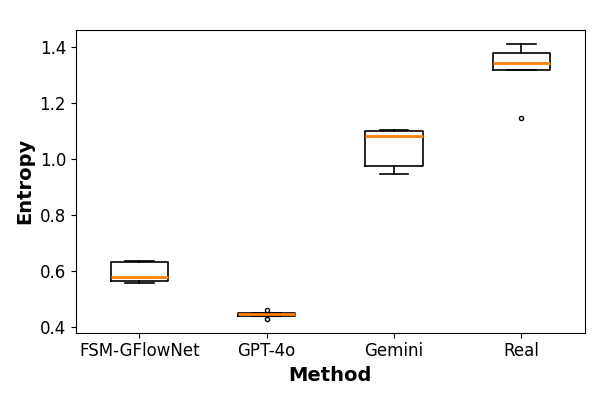}
    \caption{\textbf{(c)} Entropy}
\end{subfigure}
\hfill
\begin{subfigure}[t]{0.45\textwidth}
    \centering
    \includegraphics[width=\textwidth]{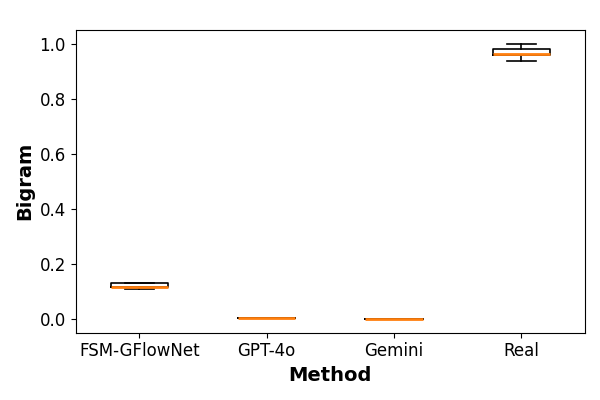}
    \caption{\textbf{(d)} Bigram Overlap}
\end{subfigure}

\caption{Distribution of evaluation metrics across methods using \textbf{real users (UIC HCI) logs as baseline}.}
\label{fig:boxplots-real-baseline}
\vspace{-0.2in}
\end{figure*}

\begin{figure*}[!t]
\centering
\label{figs-gt}
\captionsetup[subfigure]{labelformat=empty}

\begin{subfigure}[t]{0.45\textwidth}
    \centering
    \includegraphics[width=\textwidth]{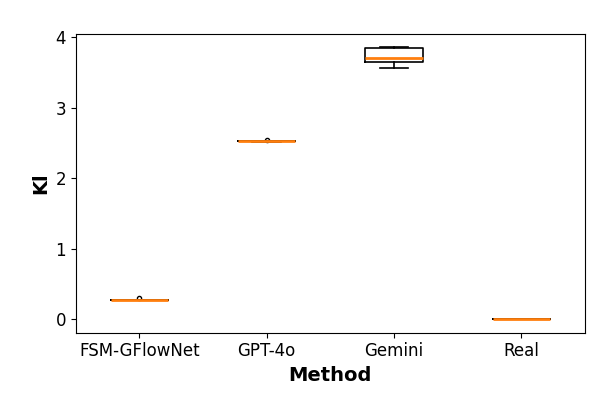}
    \caption{\textbf{(a)} KL Divergence}
\end{subfigure}
\hfill
\begin{subfigure}[t]{0.45\textwidth}
    \centering
    \includegraphics[width=\textwidth]{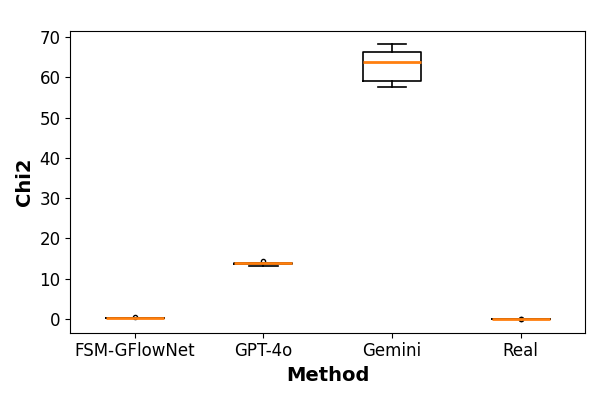}
    \caption{\textbf{(b)} Chi-Squared Distance}
\end{subfigure}

% \vspace{1mm}

\begin{subfigure}[t]{0.45\textwidth}
    \centering
    \includegraphics[width=\textwidth]{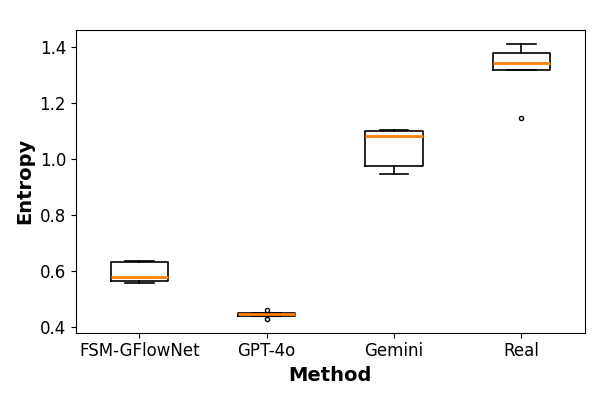}
    \caption{\textbf{(c)} Entropy}
\end{subfigure}
\hfill
\begin{subfigure}[t]{0.45\textwidth}
    \centering
    \includegraphics[width=\textwidth]{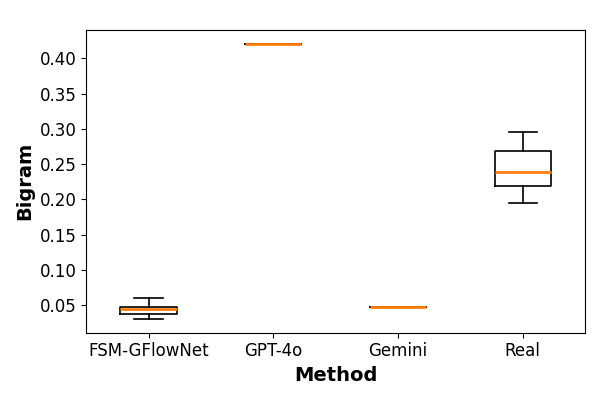}
    \caption{\textbf{(d)} Bigram Overlap}
\end{subfigure}

\caption{Distribution of evaluation metrics across methods using \textbf{ground truth (GT) logs as baseline}.}
\label{fig:boxplots-gt-baseline}
\vspace{-0.2in}
\end{figure*}

\subsection{Evaluation Metrics}

To quantitatively evaluate the distributional fidelity and structural realism of synthetic UI logs, we use four well-established statistical metrics:

\begin{itemize}
    \item \textbf{Kullback-Leibler (KL) Divergence}: Given two discrete probability distributions $P = \{p_1, \dots, p_n\}$ (baseline) and $Q = \{q_1, \dots, q_n\}$ (generated), the KL divergence is computed as:
    \begin{equation}
        \mathrm{KL}(Q \parallel P) = \sum_{i=1}^{n} q_i \log\left(\frac{q_i}{p_i + \epsilon}\right),
    \end{equation}
    where $\epsilon$ is a small constant added for numerical stability. This metric captures the divergence of event frequencies in the generated log from the baseline distribution.

    \item \textbf{Chi-Squared ($\chi^2$) Distance}: For two distributions $O = \{o_1, \dots, o_n\}$ (observed counts) and $E = \{e_1, \dots, e_n\}$ (expected counts), the Chi-squared distance is computed as:
    \begin{equation}
        \chi^2(O, E) = \sum_{i=1}^{n} \frac{(o_i - e_i)^2}{e_i + \epsilon},
    \end{equation}
    where $\epsilon$ prevents division by zero. This metric quantifies statistical deviation in event counts between the synthetic and baseline logs.

    \item \textbf{Entropy}: Entropy measures the uncertainty or diversity in the distribution of actions within a log. For a probability distribution $Q = \{q_1, \dots, q_n\}$ over actions, the entropy is:
    \begin{equation}
        \mathrm{H}(Q) = -\sum_{i=1}^{n} q_i \log(q_i + \epsilon).
    \end{equation}

    \item \textbf{Bigram Overlap}: This metric measures local structural coherence by computing the fraction of bigrams (consecutive action pairs) in the generated sequence that are also present in the baseline log:
    \begin{equation}
        \mathrm{BigramOverlap}(S_g, S_b) = \frac{|B(S_g) \cap B(S_b)|}{\max(|B(S_b)|, 1)},
    \end{equation}
    where $B(S)$ denotes the multiset of bigrams from sequence $S$, and $S_g$, $S_b$ refer to generated and baseline sequences respectively.
\end{itemize}

Each synthetic log is evaluated against two reference baselines: (i) the aggregate distribution of real user logs from the UIC HCI dataset, and (ii) a deterministic Ground Truth (GT) trajectory synthesized via GPT-4o, which serves as an optimal sequence benchmark.

To comprehensively assess fidelity across behavioral scales, we conduct four experiments: (i) aggregate-level comparison using real logs as the baseline (see Table~\ref{tab:quantitative-evaluation}(a)), (ii) aggregate-level comparison using the GT as baseline (see Table~\ref{tab:quantitative-evaluation}(b)), (iii) per-file comparison with real logs (see Figure~\ref{fig:boxplots-real-baseline}), and (iv) per-file comparison with the GT (see Figure~\ref{fig:boxplots-gt-baseline}).

\subsection{Results and Analysis}
\label{subsec:results}

\subsubsection{Evaluation with Real User Logs as Baseline}

Under the real user logs baseline (Table~\ref{tab:quantitative-evaluation}(a)), FSM-GFlowNet achieves significantly lower divergence values—\textbf{KL divergence} of $0.2769$ and \textbf{Chi-squared distance} ($\chi^2$) of $0.3522$—second only to the real logs themselves, and far outperforming GPT-4o ($2.5294$, $13.8020$) and Gemini ($3.7233$, $63.0355$). This indicates that FSM-GFlowNet preserves task-relevant action frequencies in a manner that closely mirrors authentic user behavior.

These trends are strongly corroborated by the box plots in Figure~\ref{fig:boxplots-real-baseline}, which show not only a lower median for FSM-GFlowNet across all metrics, but also a \textbf{tight interquartile range (IQR)} and \textbf{minimal outliers}, reflecting its statistical robustness and generation stability. In contrast, GPT-4o and Gemini exhibit wider IQRs and noticeable outliers in both KL and $\chi^2$, pointing to unpredictable and inconsistent outputs.

In terms of \textbf{entropy}, which captures behavioral diversity, FSM-GFlowNet produces moderately diverse outputs ($\approx 0.60$), well below real logs ($\approx 1.32$), but importantly, more diverse than GPT-4o ($0.4445$). While Gemini shows higher entropy ($\approx 1.0417$), the variance across its runs (as seen in the entropy box plot) is also significantly higher, suggesting unstructured or noisy variability. The entropy box for FSM-GFlowNet, on the other hand, remains compact and stable across runs.

Most notably, FSM-GFlowNet leads by a large margin in \textbf{bigram overlap} ($0.1214$), whereas GPT-4o ($0.0028$) and Gemini ($0.0007$) produce almost negligible overlap. The box plots clearly show that only FSM-GFlowNet maintains a positively separated distribution in this metric, capturing valid short-term transitions between actions, which is critical for modeling realistic UI sequences.

\subsubsection{Evaluation with Ground Truth (GT) Trajectory as Baseline}

Using the GT trajectory as the reference (Table~\ref{tab:quantitative-evaluation}(b)), GPT-4o performs best in terms of divergence metrics—\textbf{KL divergence} of $0.6487$ and \textbf{Chi-squared distance} of $0.9473$—which is expected given that the GT was constructed using GPT-4o prompting. FSM-GFlowNet, in contrast, yields higher divergence values (KL = $16.9481$, $\chi^2 = 6.86 \times 10^8$), since it does not attempt to replicate the GT path deterministically. However, these high divergence scores reflect \textbf{constructive deviation}, capturing richer behavioral variants rather than failure. The strength of FSM-GFlowNet becomes particularly evident in the \textbf{bigram overlap}, where it again achieves the highest average score ($0.1214$) and the \textbf{most stable distribution}, as seen in Figure~\ref{fig:boxplots-gt-baseline}. Unlike GPT-4o and Gemini, which exhibit near-zero overlap and little variability, FSM-GFlowNet preserves FSM-consistent short-horizon transitions even when diverging from the GT path. This suggests generalization beyond optimal sequences while retaining structural fidelity. In terms of \textbf{entropy}, FSM-GFlowNet shows consistent moderate diversity, while GPT-4o logs are tightly clustered with low entropy, indicating repetitive generation. Gemini once again displays high entropy but with large dispersion, signifying erratic transitions and weak alignment with the GT.

% \vspace{-0.1in}
It is also important to highlight the \textbf{asymmetric nature of the $\chi^2$ metric}, which helps explain the disproportionately high $\chi^2$ values observed when the Ground Truth (GT) trajectory is used as the baseline. Unlike real logs that represent a broader and smoother distribution over diverse user behaviors, the GT is a single optimal trace with a narrow action set and deterministic transitions. As a result, events that are absent in the GT but commonly occur in real or FSM-GFlowNet logs; for instance, the action ``M'' (mouse hover), which was omitted in the GT for efficiency, yield high residuals in the $\chi^2$ computation. Since $\chi^2$ distance penalizes mismatches based on expected frequencies, even minor deviations become amplified when the denominator (expected count) is near zero. Conversely, when real logs are used as the baseline, their natural variability provides smoother expected distributions, leading to more stable and interpretable divergence scores. Therefore, the apparent discrepancy in $\chi^2$ values under GT and real baselines stems from the metric’s statistical sensitivity to the sparsity and peakedness of the reference distribution.

\section{Use-case: Intent Classification}
\label{sec:use}

To demonstrate the practical utility of our framework beyond generative evaluation, we apply FSM-GFlowNet-generated logs to a downstream task: \textbf{intent classification}. This task evaluates whether models trained solely on synthetic event sequences can accurately infer user intent from real-world interaction logs—thus validating the semantic fidelity and behavioral realism of the generated data.

\vspace{-0.1in}
\subsection{Experimental Setup}
We use 50000 log files, each having 1000-1500 rows (events), produced from FSM-GFlowNet sampling, as training data. Real user logs from the UIC HCI dataset are reserved exclusively for testing. Each log entry is represented as a $(\texttt{state}, \texttt{event})$ pair. Using domain-informed heuristics, we define intent labels as follows:
\begin{itemize}
    \item \textbf{Open\_App}: $\texttt{event} = \texttt{A1}$ or $\texttt{state} = \texttt{S1}$
    \item \textbf{navigate}: $\texttt{event} = \texttt{A8}$ or $\texttt{state} = \texttt{S2}$
    \item \textbf{Edit}: $\texttt{event} \in \{\texttt{K1}, \texttt{K3}, \texttt{K4}\}$ or $\texttt{state} \in \{\texttt{S3}, \texttt{S4}\}$
\end{itemize}

We encode each $(\texttt{state}, \texttt{event})$ combination as a categorical feature using \texttt{CountVectorizer}, then train two standard linear classifiers—Logistic Regression and a linear-kernel Support Vector Machine (SVM). 

\vspace{-0.1in}
\subsection{Evaluation on Real Logs}

The trained classifiers are evaluated on real interaction sequences from the UIC HCI dataset. The performance results are shown in Table~\ref{tab:intent-classification-results}.

\begin{table}[ht]
\centering
\caption{Intent classification accuracy and macro F1-score}
\vspace{-0.1in}
\label{tab:intent-classification-results}
\begin{tabular}{|l|c|c|}
\hline
\textbf{Model} & \textbf{Accuracy (\%)} & \textbf{Macro F1-score} \\
\hline
Logistic Regression & 77.58 & 0.431 \\
SVM (Linear)        & 77.58 & 0.431 \\
\hline
\end{tabular}
\vspace{-0.1in}
\end{table}

\subsection{Discussion and Implications}
\vspace{-0.05in}
Despite being trained exclusively on synthetic data, both models generalize effectively to real-world UI logs, achieving approximately 78\% accuracy. This demonstrates that FSM-GFlowNet-generated logs retain meaningful semantic and behavioral patterns necessary for intent recognition.

Notably, the classifiers perform well on high-frequency intents such as \texttt{navigate} and \texttt{Open\_App}, affirming the structural alignment between synthetic and real data. While performance on underrepresented intents (e.g., \texttt{Edit}) is lower—due to real-world imbalance and potential contextual variance—the overall results highlight successful transfer of behavioral signals.

This experiment underscores the broader applicability of FSM-GFlowNet in HCI and log-based modeling scenarios. By providing structurally grounded yet diverse synthetic sequences, the framework facilitates scalable training of interaction-aware models—particularly valuable in settings with limited or privacy-constrained user data.

\vspace{-0.05in}
\section{Conclusion and Future Directions}
\label{sec:con}
\vspace{-0.05in}
We have presented a framework for structured synthetic log generation that integrates Finite State Machines (FSMs) with Generative Flow Networks (GFlowNets). By leveraging FSMs to encode domain-specific interaction rules and constraining GFlowNet sampling via dynamic action masking, the proposed approach ensures that each generated sequence is both semantically valid and behaviorally diverse. The framework is trained using a flow-matching objective with a hybrid reward that balances strict structural compliance with statistical alignment to real-world data. While this framework is instantiated and evaluated in the domain of UI interaction modeling—using FSMs derived from expert GPT-4o trajectories and real logs from the UIC HCI dataset—its modularity and symbolic design allow it to generalize across domains where discrete symbolic sequences and rule-based transitions are essential. These include cybersecurity log simulation, robotic task planning, educational activity tracing, and workflow execution modeling.

Empirical results based on distributional metrics (KL divergence, Chi-squared distance, entropy, and bigram overlap) confirm the high fidelity of our generated logs. Moreover, we demonstrate real-world applicability through a downstream task—intent classification—where classifiers trained solely on synthetic FSM-GFlowNet logs exhibit strong generalization to real user sessions. These outcomes establish FSM-GFlowNet as an effective tool for scalable behavioral simulation, training data augmentation, and sequence-aware modeling in structured interaction domains.

Despite these advantages, the framework inherits certain limitations from FSM-based modeling. Specifically, FSM construction requires domain expertise and explicit encoding of valid transitions, which may limit its out-of-the-box applicability to highly dynamic or ill-defined environments. Additionally, our current experiments focus on discrete symbolic events; extending the framework to handle continuous or multimodal sequences, such as gaze movements, speech commands, or sensor traces, is a promising direction.

Future work will explore automatic FSM induction via unsupervised learning or LLM-guided program synthesis, integration with reinforcement learning for task optimization, and multi-agent extensions of GFlowNet-based sampling for collaborative and adversarial interaction modeling. These directions aim to further enhance the adaptability, automation, and scalability of the FSM-GFlowNet paradigm.

\bibliographystyle{IEEEtran}
\bibliography{sample}

\end{document}